\documentclass[submission,copyright]{eptcs}
\providecommand{\event}{Words 2011}
\usepackage{breakurl}

\usepackage{amsmath}
\usepackage{amssymb}
\usepackage{xspace}

\def\UD{{\it UD}\xspace}
\newtheorem{definition}{Definition}[section]
\newtheorem{proposition}[definition]{Proposition}
\newtheorem{theorem}[definition]{Theorem}
\newtheorem{lemma}[definition]{Lemma}
\newtheorem{corollary}[definition]{Corollary}
\newtheorem{example}[definition]{Example}
\newtheorem{remark}[definition]{Remark}
\newenvironment{proof}{\vspace{3pt}{\noindent\em Proof.} }
                       {\hfill \hbox{$\square$}\vspace{3pt}\\}
\title{Monoids and Maximal Codes}
\author{Fabio Burderi
\institute{Universit\`{a} Degli Studi\\ Palermo, Italy}
\institute{ Dipartimento di Matematica ed
Applicazioni,\\Universit\`{a} Degli Studi di Palermo,\\
Via Archirafi 34, 90123 Palermo, Italy\\}
\email{burderi@math.unipa.it}
}
\def\titlerunning{Monoids and Maximal Codes}
\def\authorrunning{F. Burderi}

\begin{document}
\maketitle


\begin{abstract}
  In recent years codes that are not Uniquely Decipherable (\UD)
  have been studied partitioning them in classes that localize the
  ambiguities of the code. A natural question is how we can extend
  the notion of maximality to codes that are not \UD. In this paper
  we give an answer to this question.\\
  To do this we introduce a partial order in the set of
  submonoids of a monoid showing the existence, in this poset,
  of maximal elements that we call {\em full} monoids. Then a set of
  generators of a full monoid is, by definition, a maximal code. We
  show how this definition extends, in a natural way, the existing
  definition concerning \UD codes and we find a characteristic
  property of a monoid generated by a maximal \UD code.
\end{abstract}


\section{Introduction}

At the beginning, in the context of information theory,
the word \emph{code} has denoted what we call here
Uniquely Decipherable (\UD) code, that is a set of words
with the property that every concatenation of words of
the set (called \emph {message}) has
an unique decomposition in code words. This notion, in the
next years, has been weakened so we call here code just a set of
non-empty words.

A notion weaker than uniquely decipherability
has been used in several situations: to investigate natural languages
(see~\cite{Gonenc73}) or to study situations in which it is allowed
to recover the original message up to a permutation of
the code words (see~\cite{Lempel86}, \cite{Restivo89}, \cite{HeadWeber95})
or even when the only information to recover is the number
of code words (see~\cite{WeberHead96}).
In other cases the study has been oriented toward sets of words with a
constraint source (see~\cite{DalaiLeonardi05}).
In \cite{Guzman99}, Guzm{\'a}n has been introduced the notion of
{\em variety} of codes to study, in a general approach, decipherability conditions
weaker than \UD.

In \cite{BurderiRestivo07}, studying varieties of codes
under the aspect of uniform distribution of probability,
we noted that the construction,
introduced by Ehrenfeucht and Rozemberg in \cite{EhrenfeuchtRozemberg86},
for embedding a regular \UD code in a complete
and regular \UD code, also works in the ambit of
varieties of codes: the new words, introduced
by the construction, do not create new relations between
code words. Indeed the only relations between the code words are that
existing before the construction.
\\This observation has lead to deepen the study of the relations
that arise in a set of non-empty words and so in \cite{2BurderiRestivo07},
generalizing a construction used in \cite{BurderiRestivo07},
we introduced the notion of {\em coding partition}.
Roughly speaking a partition of a code is a coding partition
if any message 
has a unique factorization in blocks: a block is the
concatenation of words from one class of the partition,
and consecutive blocks are composed by words from
different classes of the partition. In this case the
possible ambiguities of the code are confined in the classes of the
partition.

In \cite{BerstelPerrinReutenauer10}, the very important 
class of maximal \UD codes is studied.
In the case of thin \UD codes, is known, for example,
the equivalence between maximality and completeness.

In this paper we define the maximality of a code
by an algebraic property of the monoid generated by
the code itself.
We show that this definition of maximality generalizes
the existing one concerning \UD codes.
\noindent We present, moreover, some classical result on \UD codes
that we can easily re-establish in the general case.


\section{Partitions of a code} \label{section.1}
Let $A$ be an alphabet. We denote by $A^*$ the set of finite
words over the alphabet $A$, and by $A^+$ the set of non-empty
finite words. $A^*$ is a monoid under the concatenation operation
of two words, with the empty word as the neutral element.
A \emph{code} $X$ is here a subset of $A^+$. Its
elements are called \emph{code words}, the elements of $X^*$ {\em
messages }.

\medskip

A code $X$ is said to be {\em uniquely decipherable} (\UD) if
every message has a unique factorization into code words, i.e. the
equality

\begin{displaymath}x_{1}x_{2} \cdots x_{n} = y_{1}y_{2} \cdots
y_{m},\end{displaymath} $x_{1},x_{2}, \dots ,x_{n},\,\,y_{1},y_{2},
\dots, y_{m}\in X$, implies\,\, $n=m$\,\, and\,\, $x_{1}=y_{1},
\dots, x_{n}=y_{n}$.

\medskip

Let $X$ be a code and let
\[
P=\{X_i \mid i \in I\}
\]
be a partition of $X$ i.e.,\;  $\bigcup_{i\in I} X_i=X$ and $X_i \cap
X_j = \emptyset$ \, iff \, $i\neq j.$

A $P$-$factorization$ of a message $w\in X^+$ is a
factorization $w=z_1z_2\cdots z_t$, where
\begin{itemize}
  \item  for each $i,\,\,\,z_i\in X_k^+$,\,\, for some $ k\in I$
  \item if $t>1$,\,\, $z_i\in X_k^+ \Rightarrow z_{i+1}\notin X_k^+$ \quad $(1\leq
  i\leq t-1)$.

\end{itemize}

The partition $P$ is called a {\em coding partition} if
any element $w\in X^+$ has a {\em unique} $P$-$factorization$,
i.e. if
\[
w= z_{1}z_{2} \cdots z_{s}= u_{1}u_{2} \cdots u_{t} ,
\]
where $z_{1}z_{2} \cdots z_{s},\,\,u_{1}u_{2} \cdots u_{t}$ are
$P$-$factorizations$ of $w$, then $s=t$ and $z_{i}=u_{i}$ for $i=
1, \dots, s$.

\noindent We observe that the trivial partition $P=\{X\}$ is always a coding partition.

Let $w\in A^+$ be a word. A \emph{factorization} of $w$ is a sequence of words
$(v_i)_{1\leq i \leq s}$ such that $w=v_1v_{2} \cdots v_{s}$. Let $X$ be a
code. A \emph{relation} is a pair of factorizations
$x_{1}x_{2} \cdots x_{s} = y_{1}y_{2} \cdots y_{t}$ into code words of
a same message $z\in X^+$; the relation is said non-trivial if the
 factorizations are distinct.
In the sequel, when no confusion arises,
sometimes we will denote by $z$ both the ``word'' $z$ and the {\em
  relation} $x_{1}x_{2} \cdots x_{s} = y_{1}y_{2} \cdots y_{t}$. We
say that the relation $x_{1}x_{2} \cdots x_{s} = y_{1}y_{2} \cdots
y_{t}$ is {\em prime} if for all $i<s$ and for all $j<t$ one has
$x_{1}x_{2} \cdots x_{i} \neq y_{1}y_{2} \cdots y_{j}$.

In \cite{2BurderiRestivo07}, the following theorem is proved.

\begin{theorem} \label{theorem.coding}
Let $P=\{X_i \mid i \in I\}$ be
a partition of a code $X$. The partition $P$ is a coding partition
iff for every prime relation $x_{1}x_{2} \cdots x_{s} =
y_{1}y_{2} \cdots y_{t}$, the code words $x_i, y_j$ belong to the
same component of the partition.
\end{theorem}

\medskip
Recall that there is a natural partial order between the partitions of a
set $X$: if $P_1$ and $P_2$ are two partitions of $X$ then $P_1\leq
P_2$ if the elements of $P_1$ are unions of elements of $P_2$
and we say that $P_2$ is {\em finer} then $P_1$.
Then from Theorem \ref{theorem.coding} we have the following corollary.

\begin{corollary}\label{corollary.coding}
Let $P$ and $P'$ be two partitions of a code $X$ with $P\leq
P'$. If $P'$ is a coding partition then $P$ also.
\end{corollary}

What follows, till Theorem \ref{decid.fini},
is stated in \cite{2BurderiRestivo07}.

\begin{theorem}\label{theorem.lattice}
The set of the coding partitions of a code $X$ is a complete
lattice.
\end{theorem}
As a consequence of previous theorem we can give the
following definition. Given a code $X$, {\em the} finest coding partition $P$ of $X$
is called the {\em characteristic} partition of $X$ and it is
denoted by $P(X)$.

A code $X$ is called {\em ambiguous} if it is not \UD.
It is called {\em totally ambiguous } ($TA$) if $|X| > 1$ and
$P(X)$ is the trivial partition: $P(X)=\{X\}$.

\begin{remark}\label{remark.UD.P(X)}
So \UD codes and $TA$ codes correspond to the two
extremal cases since a code is \UD iff $P(X)=\{\{x\}\mid x\in X\}$.
\end {remark}
\medskip

Let $X$ be a code and let $P(X)$ be the characteristic partition of
$X$. Let $X_0$ be the union of all classes of $P(X)$ having only one
element, i.e. of all classes $Z \in P(X)$ such that $|Z|= 1$. The
code $X_0$ is a \UD code and is called the {\em unambiguous
component} of $X$. From $P(X)$ one then derives another partition of
$X$
\[
P_C(X)=\{X_i \mid i \geq 0 \},
\]
where $\{X_i \mid i \geq 1\}$
is the set of classes of $P(X)$ of size greater than $1$. If there
are such sets $X_i$ with $i\geq1$, then they are $TA$.
They are called the {\em $TA$ components} of $X$.
By Corollary~\ref{corollary.coding} we have that
$P_C(X)$ is a coding partition (indeed $P_C(X)\leq P(X)$) and it is
called the {\em canonical coding partition} of $X$: it defines a
{\em canonical decomposition} of a code $X$ in at most one
unambiguous component and a (possibly empty) set  of $TA$ components.
Roughly speaking, if a code $X$ is not \UD, then its canonical
decomposition, on one hand separates the unambiguous component of
the code (if any), and, on the other, localizes the ambiguities
inside the $TA$ components of the code. On the contrary, if $X$ is
\UD, then its canonical decomposition contains only the unambiguous
component $X_0$. Moreover if $X$ is \UD then every partition of $X$
is a coding partition.

\begin{theorem}
There is a Sardinas-Patterson like
algorithm to compute the canonical coding partition of a finite
code $X$.
\end{theorem}

\begin{example}\label{example.canonical}
  Let us consider the code $X\subseteq \{0, 1\}^*$, $X = \{00, 0010, 1000, 11, 1111,
\\010, 011\}$.
In \cite{2BurderiRestivo07} it is shown that the canonical coding
partition of $X$ is $P_C(X)=\{X_0, X_1, X_2\}$ with $X_0=\{010, 011\}$, $X_1=\{00, 0010,
  1000\}$, $X_2=\{11, 1111\}$.
\end{example}

\begin{theorem}\label{decid.fini}
  Given a regular code $X$ and a partition $P=\{X_1,\dots,X_n\}$ of $X$ such that
  $X_i$, for $i=1,\dots,n$, is a regular set, it is decidable whether
  $P$ is a coding partition of $X$.
\end{theorem}

Still in \cite{2BurderiRestivo07}, it was conjectured that
{\em if $X$ is regular, the number of classes of $P_C(X)$ is finite
and each class of $P_C(X)$ is a regular set.}

\medskip

Finally, the positive answer has given in \cite{BealBurderiRestivo09} where
the following theorem and corollary are proved.

\begin{theorem} \label{theorem.regular}
The canonical partition of a regular code is finite and regular.
Its classes can be effectively computed.
\end{theorem}

\begin{corollary}
  Given a regular code $X$ and a regular partition $P=\{X_1,X_2$,
  $\dots, X_n\}$ of $X$, it is decidable whether $P$ is the canonical coding
  parition of~$X$.
\end{corollary}

From the definition of coding partition we deduce immediately
the next theorem that gives a tool to construct infinitely
many \UD codes starting from any non-$TA$ code with more
than one code word.

\begin{theorem}
Let $P=\{X_i \mid i \in I\}$ be a coding partition of a code with
$|I| > 1$.\\ Then the sets $\{X_{i_1}^+X_{i_2}^+\cdots X_{i_n}^+\mid
n\geq2,\,i_j \in I,\,i_j\neq i_{j+1}\,\,\forall\,\, 1\leq
j<n,\,i_n\neq i_1\}$ are \UD codes.
\end{theorem}

We conclude this section with the following theorem concerning the regularity
of the classes of a finite coding partition of a regular code.

\begin{theorem}
Let $\{Y_j \mid j \in J\}$ be a coding partition of a regular code $X$
and let $X_0$ be the unambiguous component of $X$.
If there exists $j_1\in J$ such that
$Y_{j_1}$ is not regular then we have $Y_{j_1}\cap X_0\neq \emptyset$.
Moreover if $J$ is finite then there exists $j_2\in J,\,j_2\neq j_1$ such that
also $Y_{j_2}$ is not regular and $Y_{j_2}\cap X_0\neq \emptyset$.
\end{theorem}

\begin{proof}
Let $P_C(X)=\{X_0,X_1,\dots, X_n\}$ be the regular and finite
canonical coding partition of $X$.
If, by contradiction, $Y_{j_1}\cap X_0 = \emptyset$ then, recalling
how $P_C(X)$ rises from $P(X)$ and recalling that $P(X)$ is
the finest coding partition of $X$, we see that $Y_{j_1}$
is a finite union of some of the regular
codes $\{X_1,\dots, X_n\}$ and so it is regular: a contradiction.
Then $Y_{j_1}\cap X_0\neq \emptyset$. Moreover if $J$ is finite
then if, by contradiction, all the $Y_j$ for  $j\neq j_1$ where
regular, then $Y_{j_1}$ where the complement, with respect to the
regular code $X$ of a regular code and so $Y_{j_1}$ where
regular against the hypothesis. Then there exists
$j_2\in J,\,j_2\neq j_1$ such that
also $Y_{j_2}$ is not regular and, by the first part of the proof,
$Y_{j_2}\cap X_0\neq \emptyset$.
\end{proof}

\begin{example}
Let $X$ be the regular \UD code $X=a^+b^+$. Then $X_0=X$ and
put $Y_1:=\{a^nb^n\mid n\geq1\},\,Y_2:=X\setminus Y_1$ we have
that $P=\{Y_1,Y_2\}$ is a coding partition of $X$ in two non-regular
classes.
\end{example}


\section{Free factorizations of a monoid} \label{section.3}

In this section the previous results are restated in an algebraic setting
making use of the free product of monoids.

Given a code $X\subseteq A^*$ we can study the properties of the monoid $M=X^*$.
On the contrary, if we start with a monoid $M\subseteq A^*$, we can study
the characteristic properties of the different sets $X\subseteq A^+$ of generators of $M$.
We recall that any submonoid $M$ of $A^*$ has a unique minimal set of
generators $X = (M\smallsetminus 1)\smallsetminus(M\smallsetminus 1)^2$, where $1$ is the empty word
(see \cite{BerstelPerrinReutenauer10}); in such a case we say that $X$ is the base
of $M$. In general we say that a code $X$ is {\em a base} if $X$ is the base of $X^*$.

\noindent It is natural to investigate how the properties of a partition
of a code are related to those of the monoids generated by the classes
of the partition.

Given a partition $P=\{X_i \mid i \in I\}$ of a code
$X\subseteq A^+$, the condition that every word $w \in X^+$ admits a
unique $P$-$factorization$ has a natural algebraic interpretation in
terms of free product of monoids.

Let $M$ be a monoid generated by submonoids
$M_\lambda, \lambda\in\Lambda$, and let $m\in M$. An expression
of $m$ of the form $\,m_1m_2\cdots m_r$, where $r\geq 0$, $1 \neq
m_i\in M_{\lambda_i}$, $\lambda_i \neq \lambda_{i+1}$,
is said in {\em reduced form} with respect to $M_\lambda$'s. By definition, $M$ is
the free product of the $M_\lambda$'s iff every element of
$M$ has an unique expression in reduced form with respect to $M_\lambda$'s and we
write $M=Fr_{\lambda\in\Lambda}\,\,M_\lambda$. In the finite case we also
write $M=M_{\lambda_1} * \cdots * M_{\lambda_n}$.

The previous results can be expressed then in the following form.
\begin{theorem}\label{theorem.free}
Let $X\subseteq A^+$ be a code, let $P=\{X_i \mid i \in I\}$ be
a partition of  $X$ and let $M=X^*$, $M_i= X_i^*$ with $i \in I$.
If $P$ is a coding partition of $X$ then $M$ is
the free product of the $M_i$'s. Conversely let $M$ be
the free product of the submonoids $M_i$'s, let $X_i$
be sets of generators of $M_i$ and let $X=\bigcup_{i\in I} X_i$.
Then $P=\{X_i \mid i \in I\}$ it is a coding partition of $X$.
\end{theorem}

It's natural at this point to introduce the notion of free
factorizations of a monoid.
\begin{definition}
A family $\{M_\lambda\, |\, \lambda\in\Lambda\}$ of submonoids of $M$
is a {\em free factorization} of $M$ if $M$ is
the free product of the $M_\lambda$'s. The $M_\lambda$'s are called
the {\em free factors} of the free factorization; moreover we say that
a monoid $M$ is {\em freely indecomposable} if $M$
cannot be expressed as a free product of nontrivial monoids.
\end{definition}

We stress that a free factor is not, in general, a free monoid.

\begin{remark}\label{remark.UD.code}
We note that a monoid $M$ is freely indecomposable iff any set of generators
of $M$ is a totally ambiguous code. From another hand we have
that a code $X$ is \UD\, iff \, $X^*=Fr_{x\in X}\, \{x\}^*$
so, in particular, the monoid $X^*$ is free.
\end{remark}

The next proposition comes directly from the definition of free product
of monoids: it is the Corollary \ref{corollary.coding} restated in
terms of monoids.

\begin{proposition}\label{partition.UD.code}
Let $M=Fr_{\lambda\in\Lambda}\,\,M_\lambda$ and let $\{\Lambda_\mu\,|\,\mu\in \Gamma\}$
be a partition of $\Lambda$. Set $\forall\mu\in\Gamma,\,\, M_{\mu}$
the monoid generated by $\{M_\lambda\,|\,\lambda \in \Lambda_\mu\}$
then $M_{\mu}=Fr_{\lambda \in \Lambda_\mu}\,M_\lambda$
and $M=Fr_{\mu \in \Gamma}\,M_{\mu}$.
\end{proposition}

Starting with an arbitrary family of submonoids of $A^*$, analogously
to what we have made with a code $X$, we can partition the family in
classes in such a way that the monoid generated by the family is
the free product of the monoids generated by each class of the
partition. On the contrary, if we have a monoid $M$, we can consider
the family of all the free factorizations of $M$ and define a partial
order on this family.

\begin{definition}
Let $F_1=\{M_\mu\,|\,\mu\in \Lambda_1\},\,\,F_2=\{M_\lambda\,|\,\lambda\in \Lambda_2\}$
be two free factorizations of a monoid $M$. We say that $F_1\leq F_2$
if there exists a partition $\{\Lambda_\mu\, |\, \mu\in \Lambda_1 \}$ of $\Lambda_2$
such that for each $\mu$, $M_\mu=Fr_{\lambda\in \Lambda_\mu}\,M_\lambda$.
\end{definition}

By Theorem \ref{theorem.lattice} and Theorem \ref{theorem.free}
we deduce the following theorem.

\begin{theorem}
Given a monoid $M$ the family of the free factorizations of $M$
is a complete lattice.
\end{theorem}

As in the case of the canonical partition of a code, {\em the} finest free factorization
of a monoid $M$ is called the {\em characteristic} free factorization of $M$ and it is
denoted by $\mathcal{F}(M)$ or, if we want to make the free factors explicit,
 $\mathcal{F}(M)=Fr_{\lambda\in\Lambda}\,\,M_\lambda$.

Now let $M_0$ be the monoid generated by all the free factors of $\mathcal{F}(M)$ having only one
generator. The monoid $M_0$ is then a free monoid and it is called the {\em free
component} of $M$. From $\mathcal{F}(M)$ one then derives another decomposition of $M$

\[
\mathcal{F}_C(M)=M_0*Fr_{\substack{\lambda\in\Lambda
}}\,\,M_\lambda,
\]
where the $M_\lambda$\,'s
are the free factors of $\mathcal{F}(M)$ having more then one generator. If there
are such monoids $M_\lambda$ then they are not free and
they are, of course, freely
indecomposable. They are called the {\em freely
indecomposable components} of $M$. By Proposition \ref{partition.UD.code} we have that
$\mathcal{F}_C(X)$ is a free factorization of $M$ (indeed $\mathcal{F}_C(M)\leq \mathcal{F}(M)$) and it is
called the {\em canonical free factorization} of $M$: it defines a
{\em canonical decomposition} of a monoid $M$ in at most one
free component and a (possibly empty) set  of freely indecomposable components.

\begin{example}\label{free.factorization.A^*}
Let $A=\{a_1, a_2, \dots \}$. Then $\mathcal{F}(A^*)=(a_1^*)*
(a_2^*)* \cdots $, and $\mathcal{F}_C(A^*)=\{A^*\}$.
Then the poset of the free
factorizations of $A^*$ are in bijection with the poset of
the alphabet $A$.
\end{example}

Already in \cite{BealBurderiRestivo09}, the following
equivalent formulation of Theorem~\ref{theorem.regular} is given.

\begin{theorem}
Any regular submonoid $M\subseteq A^*$ admits a canonical
decomposition into a free product of at most one regular free
submonoid and finitely many (possibly zero) regular freely
indecomposable submonoids.
\end{theorem}

\begin{example}
Let $A=\{a,b,c,d\}$ and let $X\subseteq A^+$ be the following
regular code: $X=a+bb+c+ad^*b+bc^*bb$. \\In \cite{BealBurderiRestivo09}
it is shown that $P_C(X)=\{X_0, X_1\}$ where $X_0=ad^+b$ and $X_1=a+ab+bb+c+bc^*bb$.
Then the canonical decomposition of the regular submonoid
$X^*$ is $X^*=(X_0^*)*(X_1^*)$.
\end{example}


\section{Full monoids and maximal codes} \label{section.4}

Using ideas of previous section we introduce a partial order
in the family of the submonoids of $A^*$. We will prove that,
in this poset, there exist maximal elements. We call this maximal
elements \emph{full} monoids and we will say that a code is \emph{maximal}
if it is the base of a full monoid. We show that this definition
of maximality extends that concerning
\UD codes and, with Theorem~\ref{UD.max.characterization}, we will give
a characterization of maximal \UD codes depending only
on the monoid they generate.

\begin{definition}

Let $M,N\subseteq A^*$ be monoids we say that $M\preceq N$
if there exists a monoid $L\subseteq A^*$ such that $N=M*L$.
\end{definition}

\begin{proposition}
The relation $\preceq$ is a partial order on the set of
submonoids of $A^*$.
\end{proposition}

\begin{proof}
We need to prove that $\preceq$ is transitive and antisymmetric.
If $L\preceq M$ and $M\preceq N$ then $\exists L',M'$ such that
$M=L*L'$ and $N=M*M'$. Then $N=(L*L')*M'=L*(L'*M')$ and so
$N\preceq L$. Now let $M\preceq N$ and $N\preceq M$ so
$M=N*N'$ and $N=M*M'$ for some monoids $M',N'$.
Then $M=M*M'*N'$ thus $M',N'$ are trivial monoids
and so $M=N$.
\end{proof}

The first question is, given a monoid $N$, if there exists
a monoid $M$ with $N\subseteq M$ and $M$ maximal with respect to
the partial order $\preceq$.
\medskip

To answer to the previous question we first prove the
following lemma.

\begin{lemma} \label{lemma.base}
Let $M=M_1*M_2$ and let $X,X_1,X_2$ be the base of $M,M_1,M_2$ respectively.
Then $X=X_1\cup X_2$.
 \end{lemma}

\begin{proof}
Since $M=M_1*M_2$ and $X_1\,,X_2$ are the bases of $M_1$ and $M_2$ respectively,
it is clear that $X_1\cup X_2$ is a set of generators of $M$.
Let, by contradiction, $X\subsetneq X_1\cup X_2$ and let
$x'\in (X_1\cup X_2)\smallsetminus X$. We can assume
that $x'\in X_1$. Since $X$ is a set of generators of $M$, $x'=x_1x_2\cdots x_n$ with $x_i\in X$.
But $x'\in M_1$ and, by the uniqueness of the reduced form with
respect to $M_1$ and $M_2$, we have $x_i\in M_1,\,\, \forall\,\, 1\leq i\leq n$,
and so $x_i\in X_1,\,\, \forall\,\, 1\leq i\leq n$.
This shows that $X_1\smallsetminus \{x'\}$ is a set of generators of $M_1$:
a contradiction.
Thus $X_1\cup X_2$ is a minimal set of generators of $M$
and we have the thesis.
\end{proof}

\noindent As an obvious generalization we have the following
\begin{corollary}
 Let $M=Fr_{\lambda\in \Lambda}\,M_\lambda$ and let
 $X_\lambda,\,\lambda\in \Lambda$ and $X$ be the bases of
 $M_\lambda,\,\lambda\in \Lambda$ and $M$ respectively.
 Then $X=\cup_{\lambda\in \Lambda}\,\,X_\lambda$.
\end{corollary}

We note that without Lemma \ref{lemma.base}, by Theorem
\ref {theorem.free}, we only say that $Y:=X_1\cup X_2$ is a
set of generators of $M$ and that $P=\{X_1, X_2\}$
is a coding partition of $Y$. Lemma \ref{lemma.base}
says that $Y$ is the base of $M$.

Now we can prove the following theorem.

\begin{theorem}
Any submonoid $M\subseteq A^*$ is contained in a submonoid $N\subseteq A^*$,
which is maximal with respect to $\preceq$ and such that $M\preceq N$.
\end{theorem}

\begin{proof}
We will make use of Zorn's lemma.
Let $\mathfrak{F}$ be the family of all the submonoids $P\subseteq A^*$,
ordered by $\preceq$, such that
$M\preceq P,\,\forall P\in \mathfrak{F}$ and let
$\{M_\lambda\,\,|\lambda\in \Lambda\}$ be a chain
in $\mathfrak{F}$. If $\lambda < \gamma$
then there exists a submonoid $H_{\lambda,\gamma}\subseteq A^*$
such that $M_\gamma = M_\lambda*H_{\lambda,\gamma}$
and so if we call $X_\gamma,\,X_\lambda$ and $X_{\lambda,\gamma}$
the bases of $M_\gamma,\, M_\lambda,\,H_{\lambda,\gamma}$
respectively, by Lemma \ref{lemma.base}, $X_\gamma = X_\lambda\cup X_{\lambda,\gamma}$
and then $X_\lambda\subsetneq X_\gamma$.
Now, $\forall\,\,\lambda\in\Lambda$, let $X_\lambda$ the base of $M_\lambda$,
$Y:=\cup_{\lambda\in \Lambda}\,\,X_\lambda$
and let $N$ be the monoid generated by $Y$. We show that $M_\lambda \preceq N,\,\,\forall\,\,\lambda\in\Lambda$.
Let $Z_\lambda:=Y\setminus X_\lambda$ and let $H_\lambda$
the submonoid generated by $Z_\lambda$. We will prove that $N=M_\lambda * H_\lambda$.
Let $m\in N$ and let us suppose, by contradiction, that
$m$ has two different expressions in reduced form with respect to
$M_\lambda, H_\lambda$ so
$m=\,m_1m_2\cdots m_r=\,m'_1m'_2\cdots m'_s$
with $r,s\geq1$. Since $N$ is generated by $Y$
then $m=y_1y_2\cdots y_h=y'_1y'_2\cdots y'_k$ for certain $y_i,y'_j\in Y$
and, since the two expressions in reduced form with respect to
$M_\lambda, H_\lambda$ are different,
$\exists y\in \{y_1, y_2,\dots, y_h, y'_1, y'_2,\dots, y'_k\}$ such that
$y\notin X_\lambda$.
Let $\lambda_1\in \Lambda$ such that $\lambda_1>\lambda$ and
$y_i,y'_j\in X_{\lambda_1},\,\,\forall \,\,i,j$.
Then $M_{\lambda_1}= M_\lambda * H_{\lambda,\lambda_1}$ for a certain
$H_{\lambda,\lambda_1}\subseteq A^*$.
Since $m,m_i,m'_j\in M_{\lambda_1},\,\,\forall \,\,i,j$, then
the two different expressions of $m$ in reduced form with respect to
$M_\lambda, H_\lambda$ are still two different expressions
in reduced form with respect to $M_\lambda,\, H_{\lambda,\lambda_1}$.
This contradiction shows that $N=M_\lambda * H_\lambda$
and thus $M_\lambda \preceq N,\,\,\forall\,\,\lambda\in\Lambda$.
Since $M\preceq M_\lambda,\,\,\forall\,\,\lambda\in\Lambda$
then $M\preceq N$ so $N\in\mathfrak{F}$ and it is a upper bound
for the chain $\{M_\lambda\,\,|\lambda\in \Lambda\}$.
Invoking Zorn's lemma we have the thesis.
\end{proof}

\begin{remark}
By Example \ref{free.factorization.A^*} we see that if $M$ is not generated by a subset
of the alphabet $A$, then the maximal monoid $N$ which
the previous theorem refers to, is properly contained in $A^*$ i.e.
$M\preceq N\subsetneq A^*$.
\end{remark}

We give now the following definition.

\begin{definition}
We say that a submonoid $M$ of $A^*$ is {\em full} if it is
maximal with respect to the partial \\order $\preceq$.
\end{definition}

\begin{remark}
From the definition we have that if $M'\subseteq M$ and
$M'$ is full then also $M$ is full.
\end{remark}

A first statement on full monoids is
given by the following proposition.

\begin{proposition} \label{max.full}
Let $M\subseteq A^*$ be a monoid. If $M$ is maximal with respect to the
inclusion order $\subseteq$ then it is full.
\end{proposition}

\begin{proof}
We will prove that if $M$ is not full then it is not maximal with
respect to the inclusion order $\subseteq$.
If $M$ is not full then there exist a monoid $N\subseteq A^*$ and a
non trivial monoid $M_1\subseteq A^*$ such that $N=M*M_1$.
Let $X$ the base of $M_1$, $x\in X$ and let $M_2$ be
the monoid $(x^2)^*$. Then we have
$M\subsetneq M*M_2\subsetneq N$.
\end{proof}

We recall that the submonoids of $A^*$ maximal with respect to the
inclusion order $\subseteq$ are ``few'': in fact it is easy to see that
a submonoid $M$ of $A^*$ is maximal with respect to the
inclusion order $\subseteq$ iff $M=A^*\smallsetminus\{a\}$ for a certain $a\in A$.

\medskip

A \UD code $X\subseteq A^+$ is said to be a {\em maximal} \UD code
if $X$ is not properly contained in any other \UD
code over $A$.

\medskip

Now we extend the notion of maximality
to codes that are not \UD.
\begin{definition}
A code $X\subseteq A^+$ is said {\em maximal}
if the monoid $X^*$ is full.
\end{definition}

The next theorem shows how this notion generalizes
that of maximality given for \UD codes.

\begin{theorem}\label{maximaUD.maximal}
Let $X$ be a \UD code. Then $X$ is a maximal \UD code iff
$X^*$ is a full monoid.
\end{theorem}

\begin{proof}
If $X$ is a maximal \UD code then $\forall w\in A^+$, $X':=X\cup \{w\}$
is not a \UD code and, by Remark \ref{remark.UD.code} and Proposition \ref{partition.UD.code},
this imply that $\forall w\in A^+$, $(X')^*$ is not the free product of $X^*$ and $\{w\}^*$
and this is true iff $X^*$ is full. 
\end{proof}

A free monoid $M\subseteq A^*$ is said \emph{maximal free} if
$M\neq A^*$ and $M$ is not properly contained in any other free monoid
different from $A^*$.

If a free monoid is maximal free then it is full. Indeed if
a free monoid is maximal free then its base is a maximal \UD code
(see \cite{BerstelPerrinReutenauer10})
so by Theorem~\ref{maximaUD.maximal} the monoid is full.

\noindent We have proved then the following theorem.
\begin{theorem}\label{maxfree.full}
Let $M$ be a free monoid. If $M$ is maximal free then it is full.
\end{theorem}

\begin{remark}

In \cite{BerstelPerrinReutenauer10} it is proved that
uniform codes $A^n$ are maximal \UD codes $\forall n\geq 1$ and
it is been underlined that with
$n=lm,\, l,m>1$, we have $(A^n)^*\subsetneq (A^m)^*\subsetneq A^*$.
This has two consequences: from one hand, by Theorem~\ref{maximaUD.maximal},
we can see that the inverse of Proposition~\ref{max.full} is false,
moreover, since the monoids $(A^n)^*$ are free,
again by Theorem~\ref{maximaUD.maximal}, also the inverse of
Theorem~\ref{maxfree.full} is false.

\end{remark}

Recalling that if $M$ is a free monoid then
its base is a \UD code, then from
Theorem~\ref{maximaUD.maximal} we have the following
characterization of a maximal \UD codes in terms of
algebraic properties of the monoid generated by the
code itself.

\begin{theorem} \label{UD.max.characterization}
Let $X\subseteq A^+$ be a code that is a base.
Then $X$ is a maximal \UD code iff $X^*$ is a full and
free submonoid of $A^*$.
\end{theorem}

We see now how with this notion of maximality
we will recover some results concerning the \UD codes.

We first recall some definitions.

\noindent A word $w\in A^*$ is a {\em factor} of a word $z\in A^*$
if there exist $u,v\in A^*$ such that $z=uwv$.
For any $X\subseteq A^*$ let $F(X)$ denote the set of factors of
words in $X$.

\noindent A set $X\subseteq A^*$ is {\em dense} if $F(X)= A^*$. A
set that is not dense is called {\em thin}.

\noindent Finally, a set $X\subseteq A^*$ is {\em complete} if $X^*$
is dense.

\begin{theorem}
Let $X\subseteq A^+$ be a maximal code then it is a complete set.
\end{theorem}

\begin{proof}
Let $X$ be a code over the alphabet $A$, with $card(A) \geq 2$
(the case $card(A) < 2$ is trivial). We will prove that if $X$
is not complete then $X^*$ is not full. If $X$ is not complete,
there exists a word $v \in A^*$ such that $v$ does not belong to
$F(X^*)$. Let $a$ be the first letter of $v$ and let $b \in
A\smallsetminus\{a\}$. Consider the word $w = vb^{|v|-1}$. By
construction, $w$ is {\em unbordered}, i.e. no proper prefix of
$w$ is a suffix of $w$. Since $v$ does not belong to $F(X^*)$, we
have that also $w$ does not belong to $F(X^*)$.
Let $M:=(X\cup \{w\})^*$ we now prove that every word $t \in (X \cup \{w\})^*$ has an
unique expression in reduced form with respect to $X^*,\,\{w\}^*$.
Indeed, since $w$ is unbordered, we
can uniquely distinguish all occurrences of $w$ in $t$, i.e. $t$
has a unique factorization of the form
\[
t = u_1wu_2w \cdots wu_n,
\]
with $n \geq 1$ and $u_i \in X^*$, for $i = 1,\dots,n$.

\noindent This shows that $M=(X^*)*(w^*)$ and $X^*$ is not full.
\end{proof}

By the previous theorem we deduce the following corollary.

\begin{corollary}
Any full monoid $M\subseteq A^*$ is dense in $A^*$.
\end{corollary}

The inverse of previous corollary is not true. Indeed the Dyck
code $D$ over $A=\{a,b\}$ is a \UD dense code and for each $x\in D$ the code
$D\smallsetminus \{x\}$ remains dense (see \cite{BerstelPerrinReutenauer10})
but it is no more a maximal \UD code and so by Theorem
\ref{UD.max.characterization} $(D\smallsetminus \{x\})^*$ it is not
full in $A^*$.

The next lemma holds (see \cite{BerstelPerrinReutenauer10}).

\begin{lemma}
Let $X\subseteq A^+$ be a thin and complete code. Then all words
$w\in A^*$ satisfy
\[
(X^*wX^*)^+\cap X^*\neq \emptyset.
\]
\end{lemma}

\noindent Then we can prove the following theorem.

\begin{theorem} \label{complete-maximal}
Let $X$ be a thin code. If $X$ is complete then it is maximal.
\end{theorem}

\begin{proof}
Let $M\subseteq A^+$ be a monoid and let $1\neq w\in M$.
By previous lemma there exist $v_1, v_2\in X^*$ and $z\in X^+$
such that $z=(v_1wv_2)^+$.
From this $z$ has not a unique expression in reduced
form with respect to $X^*$ and $M$. Then $X^*$ is full and $X$ is
a maximal code.
\end{proof}

Putting together the last two results we have:

\begin{theorem}
Let $X\subseteq A^+$ be a thin code. Then $X$ is complete iff it is maximal.
\end{theorem}

Again in \cite{BerstelPerrinReutenauer10}, the following
result is proved.
 
\begin{proposition}
Any regular \UD code is thin.
\end{proposition}

\noindent Indeed the proof of the cited result shows the following more general proposition.
\begin{proposition}
Any regular code that is a base is thin.
\end{proposition}

Then we can conclude with the following
corollary.

\begin{corollary}
Let $X\subseteq A^+$ be a regular code that is a base. Then $X$ is complete iff it is maximal.
\end{corollary}


\section{Concluding remarks}
In this paper we have given a definition of maximality
that extends the existing one for \UD codes re-establishing,
in the general case, some classical results valid for \UD
codes. At this point it is interesting to understand which,
among the deep results concerning maximal \UD codes,
can be recovered from the more general definitions of maximality
and coding partition. (We emphasize that the notion of coding partition
generalizes that of \UD code: the ``uniquely decipherability'' at the level of
classes of the partition takes the place of the uniquely decipherability
existing between the words of a \UD code.)
Two subjects that it is possible to deepen are composition of codes
and probability distributions.

\nocite{*}
\bibliographystyle{eptcs}
\bibliography{mmc}

\begin{thebibliography}{10}
\providecommand{\bibitemdeclare}[2]{}
\providecommand{\urlprefix}{Available at }
\providecommand{\url}[1]{\texttt{#1}}
\providecommand{\href}[2]{\texttt{#2}}
\providecommand{\urlalt}[2]{\href{#1}{#2}}
\providecommand{\doi}[1]{doi:\urlalt{http://dx.doi.org/#1}{#1}}
\providecommand{\bibinfo}[2]{#2}

\bibitemdeclare{article}{BealBurderiRestivo09}
\bibitem{BealBurderiRestivo09}
\bibinfo{author}{Marie-Pierre B\'eal}, \bibinfo{author}{Fabio Burderi} \&
  \bibinfo{author}{Antonio Restivo} (\bibinfo{year}{2009}):
  \emph{\bibinfo{title}{Coding partitions of regular sets}}.
\newblock {\sl \bibinfo{journal}{Inter. Jour. Alg. Comput.}}
  \bibinfo{volume}{Vol 19, No 8}(\bibinfo{number}{8}), pp.
  \bibinfo{pages}{1011--1023}, \doi{10.1142/S0218196709005457}.

\bibitemdeclare{book}{BerstelPerrinReutenauer10}
\bibitem{BerstelPerrinReutenauer10}
\bibinfo{author}{Jean Berstel}, \bibinfo{author}{Dominique Perrin} \&
  \bibinfo{author}{Christophe Reutenauer} (\bibinfo{year}{2010}):
  \emph{\bibinfo{title}{Codes and Automata}}.
\newblock {\sl \bibinfo{series}{Encyclopedia of Mathematics and its
  Applications}} \bibinfo{volume}{129}, \bibinfo{publisher}{Cambridge
  University Press}.

\bibitemdeclare{article}{2BurderiRestivo07}
\bibitem{2BurderiRestivo07}
\bibinfo{author}{Fabio Burderi} \& \bibinfo{author}{Antonio Restivo}
  (\bibinfo{year}{2007}): \emph{\bibinfo{title}{Coding partitions}}.
\newblock {\sl \bibinfo{journal}{Discret. Math. Theor. Comput. Sci.}}
  \bibinfo{volume}{Vol 9, No 2}(\bibinfo{number}{2}), pp.
  \bibinfo{pages}{227--240}.

\bibitemdeclare{article}{BurderiRestivo07}
\bibitem{BurderiRestivo07}
\bibinfo{author}{Fabio Burderi} \& \bibinfo{author}{Antonio Restivo}
  (\bibinfo{year}{2007}): \emph{\bibinfo{title}{Varieties of Codes and Kraft
  Inequality}}.
\newblock {\sl \bibinfo{journal}{Theory Comput. Systems}} \bibinfo{volume}{Vol
  40}, pp. \bibinfo{pages}{507--520}, \doi{10.1007/s00224-006-1320-0}.

\bibitemdeclare{inproceedings}{DalaiLeonardi05}
\bibitem{DalaiLeonardi05}
\bibinfo{author}{M.~Dalai} \& \bibinfo{author}{R.~Leonardi}
  (\bibinfo{year}{2005}): \emph{\bibinfo{title}{Non prefix-free codes for
  constrained sequences}}.
\newblock In: {\sl \bibinfo{booktitle}{International Symposium on Information
  Theory, 2005. ISIT 2005}}, \bibinfo{organization}{IEEE}, pp.
  \bibinfo{pages}{1534--1538}, \doi{10.1109/ISIT.2005.1523601}.

\bibitemdeclare{article}{EhrenfeuchtRozemberg86}
\bibitem{EhrenfeuchtRozemberg86}
\bibinfo{author}{A.~Ehrenfeucht} \& \bibinfo{author}{G.~Rozemberg}
  (\bibinfo{year}{1986}): \emph{\bibinfo{title}{Each regular code is included
  in a maximal regular code}}.
\newblock {\sl \bibinfo{journal}{RAIRO Inform. Theor. Appl.}}
  \bibinfo{volume}{20}, pp. \bibinfo{pages}{89--96}.

\bibitemdeclare{inproceedings}{Gonenc73}
\bibitem{Gonenc73}
\bibinfo{author}{G{\"u}ney G{\"o}nen\c{c}} (\bibinfo{year}{1973}):
  \emph{\bibinfo{title}{Unique decipherability of codes with constraints with
  application to syllabification of {T}urkish words}}.
\newblock In: {\sl \bibinfo{booktitle}{COLING 1973: Computational And
  Mathematical Linguistics: Proceedings of the International Conference on
  Computational Linguistics}}, \bibinfo{volume}{1}, pp.
  \bibinfo{pages}{183--193}.

\bibitemdeclare{article}{Guzman99}
\bibitem{Guzman99}
\bibinfo{author}{Fernando Guzm{\'a}n} (\bibinfo{year}{1999}):
  \emph{\bibinfo{title}{Decipherability of codes}}.
\newblock {\sl \bibinfo{journal}{J. Pure Appl. Algebra}}
  \bibinfo{volume}{141}(\bibinfo{number}{1}), pp. \bibinfo{pages}{13--35},
  \doi{10.1016/S0022-4049(98)00019-X}.

\bibitemdeclare{article}{HeadWeber95}
\bibitem{HeadWeber95}
\bibinfo{author}{Tom Head} \& \bibinfo{author}{Andreas Weber}
  (\bibinfo{year}{1995}): \emph{\bibinfo{title}{Deciding Multiset
  Decipherability}}.
\newblock {\sl \bibinfo{journal}{IEEE Trans. Inform. Theory}}
  \bibinfo{volume}{41}(\bibinfo{number}{1}), pp. \bibinfo{pages}{291--297},
  \doi{10.1109/18.370097}.

\bibitemdeclare{article}{Lempel86}
\bibitem{Lempel86}
\bibinfo{author}{Abraham Lempel} (\bibinfo{year}{1986}):
  \emph{\bibinfo{title}{On multiset decipherable codes}}.
\newblock {\sl \bibinfo{journal}{IEEE Trans. Inform. Theory}}
  \bibinfo{volume}{32}(\bibinfo{number}{5}), pp. \bibinfo{pages}{714--716},
  \doi{10.1109/TIT.1986.1057217}.

\bibitemdeclare{article}{Restivo89}
\bibitem{Restivo89}
\bibinfo{author}{Antonio Restivo} (\bibinfo{year}{1989}):
  \emph{\bibinfo{title}{A note on multiset decipherable codes}}.
\newblock {\sl \bibinfo{journal}{IEEE Trans. Inform. Theory}}
  \bibinfo{volume}{35}(\bibinfo{number}{3}), pp. \bibinfo{pages}{662--663},
  \doi{10.1109/18.30991}.

\bibitemdeclare{article}{WeberHead96}
\bibitem{WeberHead96}
\bibinfo{author}{Andreas Weber} \& \bibinfo{author}{Tom Head}
  (\bibinfo{year}{1996}): \emph{\bibinfo{title}{The Finest Homophonic Partition
  and Related Code Concepts}}.
\newblock {\sl \bibinfo{journal}{IEEE Trans. Inform. Theory}}
  \bibinfo{volume}{42}(\bibinfo{number}{5}), pp. \bibinfo{pages}{1569--1575},
  \doi{10.1109/18.532902}.

\end{thebibliography}
\end{document}